\documentclass[12pt]{article}

\usepackage[a4paper,margin=1in]{geometry}
\usepackage[T1]{fontenc}
\usepackage{setspace} 
\usepackage{amsmath}
\usepackage{amsthm}
\usepackage{amssymb}
\usepackage{amsfonts}
\usepackage{graphicx}
\usepackage{makecell}
\usepackage{subfigure}
\usepackage{tablefootnote}
\usepackage{bm}
\usepackage{dcolumn}
\usepackage{cite}
\usepackage{mathrsfs}
\usepackage{float}
\usepackage{multirow}
\usepackage{braket}
\usepackage{color}
\usepackage{tcolorbox}
\usepackage{algorithm, algorithmicx, algpseudocode}

\usepackage{tikz}
\usetikzlibrary{quantikz2}
\usepackage{appendix}
\usepackage[sort, numbers]{natbib}
\usepackage{authblk}
\usepackage{mathtools}
\usepackage{threeparttable}
\usepackage{booktabs}

\usepackage[colorlinks,citecolor=blue,linkcolor=blue,  urlcolor=blue]{hyperref}
\usepackage{orcidlink}
\usepackage[capitalize,noabbrev]{cleveref}

\DeclarePairedDelimiter\rbra{\lparen}{\rparen}
\DeclarePairedDelimiter\sbra{\lbrack}{\rbrack}
\DeclarePairedDelimiter\cbra{\{}{\}}
\DeclarePairedDelimiter\abs{\lvert}{\rvert}
\DeclarePairedDelimiter\Abs{\lVert}{\rVert}
\DeclarePairedDelimiter\ceil{\lceil}{\rceil}
\DeclarePairedDelimiter\floor{\lfloor}{\rfloor}

\newtheorem{theorem}{Theorem}
\newtheorem{lemma}[theorem]{Lemma}

\newtheorem{corollary}[theorem]{Corollary}

\newtheorem{definition}{Definition}

\newtheorem{problem}{Problem}

\newtheoremstyle{restate}{}{}{\itshape}{}{\bfseries}{.}{.5em}{\thmnote{#3}}
\theoremstyle{restate}

\newcommand{\alg}[1]{\textup{\textsf{#1}}}

\DeclarePairedDelimiterX{\pt}[1](){#1} 
\DeclarePairedDelimiterX{\bc}[1][]{#1} 

\title{Succinct quantum testers for closeness and $k$-wise uniformity of probability distributions}

\author[,a\,\orcidlink{0009-0003-7757-9858}]{Jingquan Luo\thanks{Email: luojq25@mail2.sysu.edu.cn}}
\author[,b\,\orcidlink{0000-0001-5107-8279}]{Qisheng Wang\thanks{Email: QishengWang1994@gmail.com}}
\author[,a,c\,\orcidlink{0000-0001-5941-7036}]{Lvzhou Li\thanks{Email: lilvzh@mail.sysu.edu.cn (corresponding author).}}
\affil[a]{Institute of Quantum Computing and Software, School of Computer Science and Engineering, Sun Yat-sen University, Guangzhou 510006, China}
\affil[b]{Graduate School of Mathematics, Nagoya University, Nagoya 464-8602, Japan}
\affil[c]{Quantum Science Center of Guangdong-Hong Kong-Macao Greater Bay Area (Guangdong),  Shenzhen 518045, China}

\begin{document}
\maketitle

\begin{abstract}
We explore potential quantum speedups for the fundamental problem of testing the properties of closeness and $k$-wise uniformity of probability distributions.

\begin{itemize}
    \item \textit{Closeness testing} is the problem of distinguishing whether two $n$-dimensional distributions are identical or at least $\varepsilon$-far in $\ell^1$- or $\ell^2$-distance. We show that the quantum query complexities for $\ell^1$- and $\ell^2$-closeness testing are $O\rbra{\sqrt{n}/\varepsilon}$ and $O\rbra{1/\varepsilon}$, respectively, both of which achieve optimal dependence on $\varepsilon$, improving the prior best results of \hyperlink{cite.gilyen2019distributional}{Gily{\'e}n and Li~(2020)}. 

    \item \textit{$k$-wise uniformity testing} is the problem of distinguishing whether a distribution over $\cbra{0, 1}^n$ is uniform when restricted to any $k$ coordinates or $\varepsilon$-far from any such distribution. We propose the first quantum algorithm for this problem with query complexity $O\rbra{\sqrt{n^k}/\varepsilon}$, achieving a quadratic speedup over the state-of-the-art classical algorithm with sample complexity $O\rbra{n^k/\varepsilon^2}$ by \hyperlink{cite.o2018closeness}{O'Donnell and Zhao (2018)}.
    Moreover, when $k = 2$ our quantum algorithm outperforms any classical one because of the classical lower bound $\Omega\rbra{n/\varepsilon^2}$. 
\end{itemize}

All our quantum algorithms are fairly simple and time-efficient, using only basic quantum subroutines such as amplitude estimation.
\end{abstract}

\section{Introduction}

Property testing is a fundamental problem in theoretical computer science, with  the goal being to determine whether the target object has a certain property, under the promise that the object either has the property or is ``far'' from having that property. Quantum computing has a positive impact on many problems~\cite{montanaro2016quantum, zhang2022brief}. In particular, there have been a series of works on the topic of ``quantum property testing'', and the readers can refer to~\cite{montanaro2013survey}. 

A fundamental problem in statistics and learning theory is to test properties of distributions. A few works have shown that quantum computing can speed up testing  properties of distributions, for example, see~\cite{bravyi2011quantum, montanaro2015quantum, gilyen2019distributional}. 
In this article, we consider two different distribution property testing problems. The first one is to test whether two unknown distributions are close or far enough, under the metric of $\ell^1$- or $\ell^2$-distance. Closeness testing problem is a highly-representative problem in the realm of distribution property testing~\cite{batu2000testing,diakonikolas2016new, valiant2017automatic}, and the algorithms for $\ell^2$-closeness testing are the foundation of other general distribution property testing problems~\cite{diakonikolas2016new}.  Chan~\textit{et~al.}~\cite{chan2014optimal} showed that the sample complexities for $\ell^1$- and $\ell^2$-closeness testing are $\Theta(\max(\frac{n^{2/3}}{\varepsilon^{4/3}}, \frac{n^{1/2}}{\varepsilon^2}))$ and $\Theta(\frac{1}{\varepsilon^2})$, respectively, which were further investigated by Diakonikolas and Kane~\cite{diakonikolas2016new}. 
The quantum query complexities for $\ell^1$- and $\ell^2$-closeness testing were first studied by Bravyi~\textit{et~al.}~\cite{bravyi2011quantum}, and the state-of-the-art results for the both are $O(\frac{\sqrt{n}}{\varepsilon}\log^3(\frac{{n}}{\varepsilon})\log\log(\frac{{n}}{\varepsilon}))$ and $O(\frac{1}{\varepsilon}\log^3(\frac{1}{\varepsilon})\log\log(\frac{1}{\varepsilon}))$, respectively, due to Gily{\'e}n and Li~\cite{gilyen2019distributional}.

The second problem we consider is to test whether a given distribution is $k$-wise uniform. We say that a probability distribution over $\{0, 1\}^n$ is $k$-wise uniform if its marginal distribution on every subset of $k$ coordinates is the uniform distribution (see \cref{def:k-wise-uniform} for the formal definition of $k$-wise uniformity). 
The study of $k$-wise uniformity dates back to the work of  Rao~\cite{rao1947factorial} who studied $k$-wise uniform sets, which are special cases of $k$-wise uniform distribution. Since then, $k$-wise uniformity has turned  out to be an essential tool in theoretical computer science, especially for derandomization~\cite{alon1986fast, chor1985bit}. In quantum computing, $k$-wise uniformity has also been employed to design quantum algorithms~\cite{ambainis2007quantum}.

\subsection{Query-Access Models and Problem Formulations}
\label{models-problems}

In this section, we firstly give the  definitions of classical and quantum access models for distributions. Denote the sample space by $\Omega$. For the $\ell^\alpha$-closeness testing problem,  the sample space is $\Omega = \{1, 2, \dots, n\}. $ For the $k$-wise uniformity testing problem, the sample space would be $\Omega=\{0, 1\}^n$. Note that the notation $\Omega$ is used to represent both sample space and lower bounds, whose meaning will be clear from the context.  

Consider distribution $p$ with sample space $\Omega$ which is a function $p \colon \Omega \rightarrow [0, 1]$ such that $\sum_{i \in \Omega} p_i = 1$, where the notation $p_i$ denotes the probability of element $i$. The $\ell^\alpha$-distance between two distributions $p$ and $q$ is defined as $\| p - q \|_\alpha := (\sum_{i \in \Omega} | p_i - q_i |^\alpha)^{1/\alpha}$. The total variation distance $d_{\text{TV}}(p, q)$ between $p$ and $q$ is defined as $d_{\text{TV}}(p, q) := \frac{1}{2} \|p - q\|_1$.

To get access to a classical distribution, the most natural model is sampling.

\begin{definition}[Sampling]
    A classical distribution $p$ {with sample space $\Omega$} is accessible via classical sampling if we can request samples from the distribution, i.e., get a random ${i \in \Omega}$ with probability $p_i$.  
\end{definition}

{

Several quantum query-access models for classical distributions have been proposed. We first introduce the most general and natural one, named the \textit{purified quantum query-access model}, and later compare it to others. This model has been studied in~\cite{montanaro2015quantum, gilyen2019distributional, hamoudi2018quantum, belovs2019quantum, li2022unified} and is adopted throughout this paper. The model is particularly general because it is easily interchangeable with a quantum subroutine, as discussed in~\cite{montanaro2015quantum, belovs2019quantum}. It is worth noting that the purified quantum query-access model is the weakest one among all the models discussed below. 
In other words, the query complexity upper bounds in the purified quantum query-access model actually imply the same results in other models.

\begin{definition}[Purified Quantum Query-Access]\label{purified quantum query-access}
    A classical distribution $p$ {with sample space $\Omega$} has purified quantum query-access if we have access to a unitary oracle $U_p$ (and its inverse and  controlled versions) acting as 
    
    \begin{equation*}
        U_p | 0 \rangle_A | 0 \rangle_B = \sum_{i \in \Omega} \sqrt{p_i} | \phi_i \rangle_A | i \rangle_B,
    \end{equation*}
    where $\{| \phi_i \rangle \colon i \in \Omega \}$ are some arbitrary normalized quantum states such that $\langle \phi_i | \phi_j \rangle = \delta_{ij}$.
\end{definition}

In the above definition we assume that the inverse of the oracle is accessible, which is a common assumption in quantum algorithms, especially in those based on quantum amplitude amplification
and estimation~\cite{brassard2002quantum}. This assumption is reasonable since the inverse of the oracle can be simply constructed by reversing the order of the gates in the circuit description of the oracle and conjugating them. 

The second model is the pure-state preparation access model adopted in~\cite{arunachalam2018optimal, aharonov2007adiabatic, atici2005improved, bshouty1995learning}.
\begin{definition}[Pure-State Preparation Access]
\label{pspa}
    A classical distribution $p$ with sample space $\Omega$ is accessible via pure-state preparation access model if we have access to a unitary $U_p$ (and its inverse and  controlled versions) such that 
    \begin{align*}
        U_p |0\rangle = \sum_{i \in \Omega}\sqrt{p_i}|i\rangle.
    \end{align*}
\end{definition}
The pure-state preparation access model defined in \cref{pspa} can be tranformed to the purified quantum query-access model defined in \cref{purified quantum query-access}, since we can use one query to prepare $\sum_{i \in \Omega} \sqrt{p_i} \ket{i}\ket{0}$ and then apply CNOT gates to get $\sum_{i \in \Omega} \sqrt{p_i} \ket{i}\ket{i}$, which satisfies the requirement in \cref{purified quantum query-access}.

Lastly, the third model is to encode the distribution as a long string and represnet probabilities by frequencies, which is adopted in~\cite{bravyi2011quantum, chakraborty2010new, li2018quantum, montanaro2015quantum2}. 
\begin{definition}[Discrete Quantum Query-Access]
\label{def:dqqa}
A classical distribution $p$ with sample space $\Omega$ is accessible via discrete quantum query-access if we have a quantum oracle $O_f$ (and its inverse and  controlled versions) to a function $f\colon \sbra{n} \rightarrow \Omega$ such that 
\begin{align*}
    O_f \ket{j} \ket{0} = \ket{j} \ket{f(j)} 
\end{align*}
for all $j \in \sbra{n}$, and for all $i \in \Omega$,
\begin{align*}
    p_i = \frac{\abs{\cbra{j \colon f(j) = i}}}{n}.
\end{align*}
\end{definition}
One downside of the discrete quantum query-access model is that the probabilities $p_i$ must be multiples of $1/n$. 
Also, this model can be tranformed to the purified quantum query-access model defined in \cref{purified quantum query-access}, since we can apply $O_f$ to the uniform superposition over $\sbra{n}$ to get a state satisfying the requirement in \cref{purified quantum query-access}. 
}

The first problem we consider in this article is to test whether two distributions are close or far enough, which is formally described as follows:

\begin{problem}[{$\ell^\alpha$}-Closeness Testing]
    Given $\varepsilon > 0$ and access to oracles generating two probability distributions $p$, $q$ with sample space $[n]$, and promised that $p = q$ or $\| p - q \|_{\alpha} \geq \varepsilon$, $\ell^\alpha$-closeness testing requires to  decide which the case is, with high probability, say $2/3$. The problem is called tolerant  testing,  when the promise is $\| p - q \|_{\alpha} \leq (1-\nu)\varepsilon$ or $\| p - q \|_{\alpha} \geq \varepsilon$ for $\nu \in (0, 1] $ and one is asked to decide which the case is.
\end{problem}

The second problem we consider is to test whether a distribution is $k$-wise uniform or far from any $k$-wise uniform distribution under the metric of the total variation distance. In this problem, the sample space $\Omega$ with size $2^n$ of the distribution $p$ is $\{0, 1\}^n$. 

\begin{definition}[$k$-Wise Uniformity]\label{def:k-wise-uniform}A probability distribution $p$ over $\{0, 1\}^n$ is said to be $k$-wise uniform if its marginal distribution on every subset $ S \subseteq [n] $ of $k$ coordinates is the uniform distribution, i.e., 
\[
\Pr_{x \sim p} [x_S = v] = \frac{1}{2^k}
\]
for any $v \in \{0, 1\}^k$, where $x_S$ denotes the subsequence of $x$ obtained by restricting $x$ to all the coordinates in $S$ and $x \sim p$ denotes that $x$ is drawn according to the distribution $p$.
\end{definition}

\begin{problem}[$k$-Wise Uniformity Testing]
    Given a probability distribution $p$ over $\{0, 1\}^n$ with its purified quantum query-access oracle, and a real number $\varepsilon > 0$, and promised that $p$ is $k$-wise uniform or is $\varepsilon$-far from any $k$-wise uniform distribution under the metric of the total variation distance, $k$-wise uniformity testing requires to decide which the case is, with high probability.
\end{problem}

\subsection{Contributions}

In this article, we design quantum algorithms for  {$\ell^\alpha$}-closeness testing and $k$-wise uniformity testing, respectively, using as few queries  as possible to the oracle  defined in Definition \ref{purified quantum query-access}.
The number of the queries to $U_p$, $U_p^\dagger$ and their controlled versions is called the {\it quantum query complexity} of the algorithms.  One of our main contributions is that we give the  currently best quantum algorithms, i.e., Theorem \ref{thm:l2-tester} and its corollaries, for the $\ell^\alpha$-closeness testing problem ($\alpha \in \cbra{1, 2}$).

\begin{theorem}[{\normalfont Tolerant $\ell^2$-Closeness Tester}]\label{thm:l2-tester}
    There is a quantum tester for tolerant $\ell^2$-closeness testing with query complexity $O(\frac{1}{\nu\varepsilon})$.
\end{theorem}

\begin{corollary}[{\normalfont $\ell^2$-Closeness Tester}]\label{corollary1}
    There is a quantum tester for $\ell^2$-closeness testing with query complexity $O(\frac{1}{\varepsilon})$.
\end{corollary}

\begin{proof}
    It directly follows from \cref{thm:l2-tester} with, e.g., $\nu = 1/2$.
\end{proof}

\begin{corollary}[{\normalfont $\ell^1$-Closeness Tester}]\label{corollary2}
    There is a quantum tester for $\ell^1$-closeness testing with query complexity $O(\frac{\sqrt{n}}{\varepsilon})$.
\end{corollary}

\begin{proof}
    It directly follows from the Cauchy-Schwartz inequality $\| p - q \|_2 \geq \frac{1}{\sqrt{n}} \| p 
 - q \|_1$ and taking $\varepsilon \leftarrow \frac{\varepsilon}{\sqrt{n}}$ in \cref{corollary1}.
\end{proof}

The  comparison of our results with previous classical and quantum results for the $\ell^\alpha$-closeness testing problem is presented in Table \ref{table1}. We propose quantum algorithms for $\ell^1$-closeness testing and $\ell^2$-closeness testing with query complexities $O(\frac{\sqrt{n}}{\varepsilon})$ and $O(\frac{1}{\varepsilon})$, respectively, improving the prior best results given in~\cite{gilyen2019distributional}. Compared to those algorithms that are relied on the quantum singular value transformation (QSVT)~\cite{gilyen2019quantum} framework, our algorithms not only remove the polylogarithmic factors but also is more concise.

Moreover, we also show that our quantum $\ell^1$-closeness tester and $\ell^2$-closeness tester both achieve optimal dependence on $\varepsilon$ in \cref{lb:closeness}.

\begin{theorem}[{\normalfont Lower Bound for the $\ell^\alpha$-Closeness Testing Problems}]
\label{lb:closeness}
    For $\alpha \in \cbra{1, 2}$, the quantum query complexity of $\ell^\alpha$-closeness testing  is $\Omega(\frac{1}{\varepsilon})$.
\end{theorem}

\begin{table}[htbp]
\renewcommand\arraystretch{1.5} 
\centering  

\begin{threeparttable}
\caption{ Summary of sample and query complexity results of $\ell^\alpha$-closeness.}
\label{table1}   

\begin{tabular}{cccc}   
\toprule  

Type & Reference &$\ell^1$-Closeness Testing & $\ell^2$-Closeness Testing \\ 

\midrule

Classical & \makecell{Chan~\textit{et~al.}~\cite{chan2014optimal}} & $ \Theta(\max(\frac{n^{2/3}}{\varepsilon^{4/3}}, \frac{n^{1/2}}{\varepsilon^2})) $   & $\Theta(\frac{\max(\| p \|_2, \| q \|_2)}{\varepsilon^2})$   \\

\midrule  

\multirow{4}{*}{Quantum} & Bravyi~\textit{et~al.}~\cite{bravyi2011quantum}\tnote{1}  & $ O(\frac{\sqrt{n}}{\varepsilon^8}) $ & $/$ \\

~ & Montanaro~\cite{montanaro2015quantum}\tnote{1} & $ O(\frac{\sqrt{n}}{\varepsilon^{2.5}}\log(\frac{1}{\varepsilon})) $ & $/$ \\

~ & Gily{\'e}n and Li~\cite{gilyen2019distributional} & $ O(\frac{\sqrt{n}}{\varepsilon}\log^3(\frac{{n}}{\varepsilon})\log\log(\frac{{n}}{\varepsilon})) $ & $O(\frac{1}{\varepsilon}\log^3(\frac{1}{\varepsilon})\log\log(\frac{1}{\varepsilon}))$ \\

~ & \multirow{2}{*}{This work} & {{$\bm{O(\frac{\sqrt{n}}{\varepsilon})}$, $\bm{\Omega(\frac{1}{\varepsilon})}$ }}  & {{$\bm{\Theta(\frac{1}{\varepsilon})}$ }} \\
~ & ~ & (\cref{corollary2}, \cref{lb:closeness}) & (\cref{corollary1}, \cref{lb:closeness}) \\

\bottomrule
\end{tabular}  
\begin{tablenotes}
    \item [1] { Bravyi~\textit{et~al.}~\cite{bravyi2011quantum} and Montanaro~\cite{montanaro2015quantum} actually considered the problem of estimating the total variation distance between probability distributions, which is harder than $\ell^1$-closeness testing.} 
\end{tablenotes}
\end{threeparttable}
\end{table}

The second contribution is that we give, to the best of our knowledge, the first quantum tester for the $k$-wise uniformity testing problem. The  comparison of our results with previous classical results is presented in \cref{table:k-wise}. 

\begin{theorem}[{\normalfont $k$-Wise Uniformity Tester}]\label{thm-k-uniform}
    Let $k \geq 1$ be a constant. Then there is a quantum tester for $k$-wise uniformity testing with query complexity $O(\frac{\sqrt{n^k}}{\varepsilon})$.    
\end{theorem}

\begin{table}[htbp]
\renewcommand\arraystretch{1.5} 
\centering
\caption{ Summary of sample and query complexity results of $k$-wise uniformity testing.}
\label{table:k-wise} 
\begin{threeparttable}
\begin{tabular}{ccc}   
    \toprule   
    Type & Upper bound & Lower bound \\ 
    \midrule
    \multirow{2}{*}{Classical} & \multirow{2}{*}{$O(\frac{n^{k}}{\varepsilon^2})$  \cite{o2018closeness}} & {{$\Omega(\frac{n^{(k-1)/2}}{\varepsilon})$ for $k > 2$  \cite{alon2007testing}, }}\\
    ~ & ~ & $\Omega(\frac{n}{\varepsilon^2})$ for $k = 2$  \cite{o2018closeness} \\
    \midrule 
    \multirow{2}{*}{Quantum} & {$\bm{O(\frac{\sqrt{n^k}}{\varepsilon})}$}  & {$\bm{\Omega(\frac{n^{(k-1)/3}}{\varepsilon^{2/3}})}$ for $k > 2$}  \\
    ~ & (\cref{thm-k-uniform}) & (\cref{lb:k-wise}) \\
    \bottomrule
\end{tabular} 
\end{threeparttable}
\end{table}

As shown in \cref{table:k-wise}, the quantum tester given in \cref{thm-k-uniform} with query complexity $O(\frac{\sqrt{n^k}}{\varepsilon})$ achieves a quadratic speedup over the state-of-the-art classical algorithm with complexity $O(\frac{n^{k}}{\varepsilon^2})$~\cite{o2018closeness}. 
In particular, for $k = 2$, our quantum algorithm outperforms any classical algorithm, since our quantum upper bound is $O(\frac{n}{\varepsilon})$ and the classical lower bound is $\Omega(\frac{n}{\varepsilon^2})$.

We also  prove a quantum lower bound for the $k$-wise uniformity testing problem in \cref{lb:k-wise}. However, there is still a gap between the lower and upper bounds
, both for the classical and quantum settings.

\begin{theorem}[{\normalfont Lower Bound for the $k$-Wise Uniformity Testing Problem}]
\label{lb:k-wise}
    Let $k > 2$ be a constant, $n$ be a sufficiently large integer,  and $\varepsilon \in (0, 0.228)$ be a sufficiently small  real number with $0.228^2\frac{M^n_k}{n\varepsilon^2} < 2^{n^{1/3}}$. Then the quantum query complexity of $k$-wise uniformity testing is $\Omega\rbra{\frac{n^{(k-1)/3}}{\varepsilon^{2/3}}}$.
\end{theorem}

\subsection{Techniques}

Motivated by Montanaro~\cite{montanaro2015quantum} and Gily{\'e}n and Li~\cite{gilyen2019distributional}, the general recipe is to encode the quantity (e.g. $\Abs{p - q}_2$) that we want to estimate in the amplitude of a quantum state by an elaborately constructed unitary operator, and the apply {amplitude estimation}~\cite{brassard2002quantum} to estimate it up to enough accuracy. The encoding part is the most technical part.

For the $\ell^2$-closness testing problem, to encode the distance $\Abs{p - q}_2$, Gily{\'e}n and Li~\cite{gilyen2019distributional} relied on the technique of quantum singular value transformation (QSVT)~\cite{gilyen2019quantum}, which is somewhat involved. On the contrary, we give a much more concise solution. Lemma \ref{lemma:1} plays a key role in our algorithms and it is worth noting that it is not an application of Lemma 45 of~\cite{gilyen2019quantum}, which block-encodes a density matrix, and is the premise to employ the QSVT~\cite{gilyen2019quantum} framework. As for the lower bound, we reduce the distribution discrimination problem~\cite{belovs2019quantum} to the closeness testing problem.

{For the $k$-wise uniformity testing problem, it has been shown that this property is highly related to the Fourier coefficients of the corresponding density function on index subsets $S$ such that $1 \leq | S | \leq k$~\cite{alon2007testing, o2018closeness}. 
In this article, we propose a circuit to compute the Fourier coefficients, which allows us to test whether a distribution is $k$-wise uniform or not. As for the lower bound, we reduce the collision problem~\cite{aaronson2004quantum, ambainis2005polynomial, kutin2005quantum} to the $k$-wise uniformity testing problem.}

\subsection{Related Works}

\paragraph{Classical distribution property testing} In the realm of classical computing, distribution property testing~\cite{batu2000testing, batu2013testing} has attracted a lot of interest and developed into a 
 mature research field. $\ell^\alpha$-closeness testing has been playing a key role in distribution property testing~\cite{batu2000testing, diakonikolas2016new}. For example,  Diakonikolas
and Kane~\cite{diakonikolas2016new} proposed a reduction-based method that obtains optimal testers for various problems under the $\ell^1$-norm (and other metrics),
by applying a randomized transformation to a basic $\ell^2$-closeness tester. For $\ell^1$-closeness testing,   Batu~\textit{et~al.}~\cite{batu2013testing} gave a sublinear algorithm using $\widetilde{O}(\frac{n^{2/3}}{\varepsilon^{8/3}})$ samples to $p$ and $q$.   Chan~\textit{et~al.}~\cite{chan2014optimal} achieved the optimal sample complexity $\Theta(\max(\frac{n^{2/3}}{\varepsilon^{4/3}}, \frac{n^{1/2}}{\varepsilon^2}))$, and also gave the optimal sample complexity $\Theta(\frac{1}{\varepsilon^2})$ for $\ell^2$-closeness testing. The reduction-based framework proposed by  Diakonikolas and
Kane~\textit{et~al.}~\cite{diakonikolas2016new} recovered many results in the distribution property testing field, including all the results mentioned above. See~\cite{rubinfeld2012taming, canonne2020survey} for two recent surveys.

 For the $k$-wise uniformity testing problem, although  Alon~\textit{et~al.}~\cite{alon2003almost} did not explicitly consider this problem, an upper bound $\widetilde{O}(\frac{n^{2k}}{\varepsilon^2})$ can be derived from theorems in that paper.  Alon~\textit{et~al.}~\cite{alon2007testing} improved the bound to $O(\frac{n^{k}\log^{k+1}(n)}{\varepsilon^2})$, which was further improved by  O'Donnell and Zhao~\cite{o2018closeness} to $O(\frac{n^{k}}{\varepsilon^2})$.
 
\paragraph{Quantum distribution property testing} The topic of reducing the complexity of property testing problems with the help of quantum computing has attracted a lot of research, and one can refer to~\cite{montanaro2013survey} for more details. Here we briefly describe some results on distribution  testing. The first work was due to  Bravyi~\textit{et~al.}~\cite{bravyi2011quantum} who considered three different problems including $\ell^1$-closeness testing, uniformity testing\footnote{It means testing whether an unknown distribution is close to the uniform distribution.} and orthogonality testing, and gave quantum upper bounds $O(\frac{\sqrt{n}}{\varepsilon^8})$, ${O}(n^{1/3})$(under the assumption that $\varepsilon$ is a constant), $O(\frac{n^{1/3}}{\varepsilon})$, respectively.  Chakraborty~\textit{et~al.}~\cite{chakraborty2010new} independently gave an upper bound ${O}(\frac{n^{1/3}}{\varepsilon^2})$ for uniformity testing, and further showed that identity testing to any fixed distribution can be done with query complexity $\widetilde{O}(\frac{n^{1/3}}{\varepsilon^5})$. Both Bravyi~\textit{et~al.}~\cite{bravyi2011quantum} and  Chakraborty~\textit{et~al.}~\cite{chakraborty2010new} showed that $\Omega(n^{1/3})$ quantum queries are  necessary for uniformity testing, following from a reduction from the collision problem~\cite{aaronson2004quantum, ambainis2005polynomial, kutin2005quantum}.  Montanaro~\textit{et~al.}~\cite{montanaro2015quantum} improved the $\varepsilon$-dependence of $\ell^1$-closeness testing to $\widetilde{O}(\frac{\sqrt{n}}{\varepsilon^{2.5}})$, and the upper bound is further improved to $\widetilde{O}(\frac{\sqrt{n}}{\varepsilon})$ by  Gily{\'e}n and Li~\cite{gilyen2019distributional}. 

{Apart from testing probability distributions, quantum state property testing is also an emerging topic in quantum computing, where quantum states can be seen as a generalization of probability distributions.
The closeness testing of quantum states with respect to trace distance and fidelity was investigated in a series of papers~\cite{BOW,gilyen2019distributional,wang2023quantum,wang2022new,gilyen2022improved,wang2023fast}.
}

\section{Preliminaries}

We use $[n]$ to denote the set $\{1, 2, \cdots, n\}$. Let $M^n_k = \sum_{i=1}^k \tbinom{n}{i}$. For a  quantum system composed of subsystems $A$, $B$ and a unitary operator $U$, we use $(U)_A$ to denote that the unitary $U$ is operated on the subsystem $A$, while the other subsystem remains unchanged. 

\subsection{Quantum Subroutine}

We would need a well-known quantum algorithm to estimate a probability $p$ quadratically more efficiently than  classical sampling:

\begin{theorem}[Amplitude Estimation~\cite{brassard2002quantum}]\label{amplitude estimation}

    Given a unitary ${U}$ and an orthogonal projector $\Pi$ such that $U | 0 \rangle = \sqrt{p} | \phi 
    \rangle + \sqrt{1-p} | \phi^{\perp} \rangle$ for some $p \in [0, 1]$, $\Pi | \phi \rangle = | \phi \rangle$ and $ \Pi | \phi^{\perp} \rangle = 0$, for integer $t > 0$, there is a quantum algorithm, denoted by \alg{Amplitude\mbox{-}Estimation}$(U, \Pi, t)$, outputting $\widetilde{p}$ such that 
    \begin{equation*}
        | \widetilde{p} - p | \leq 2\pi \frac{\sqrt{p(1-p)}}{t} + \frac{\pi^2}{t^2}
    \end{equation*}
    with probability at least $8/\pi^2$, using $t$ queries to $U$, $U^\dagger$ and $I - 2\Pi$. Moreover, if $p = 0$, \alg{Amplitude\mbox{-}Estimation}$(U, \Pi, t)$ outputs $\widetilde{p} = 0$ with certainty.
 
\end{theorem}

 The following corollary, which is implicitly given in~\cite{chakraborty2010new},  is an application of the above theorem.
 
\begin{corollary}\label{corollary3}
    Given a unitary $U$ and an orthogonal projector $\Pi$ such that $U | 0 \rangle = \sqrt{p} | \phi \rangle + \sqrt{1-p} | \phi^{\perp} \rangle$ for some $p \in [0, 1]$, $\Pi | \phi \rangle = | \phi \rangle$ and $ \Pi | \phi^{\perp} \rangle = 0$, for $\varepsilon \in (0, 1)$, there is a quantum algorithm, denoted by \alg{Zero\mbox{-}Tester}$(U, \Pi, \varepsilon)$, which outputs \emph{YES} if $p = 0$, or outputs \emph{NO} if $p > \varepsilon$, with probability at least $2/3$, using $O(\frac{1}{\sqrt{\varepsilon}})$ queries to $U$, $U^\dagger$ and $I - 2\Pi$. 
\end{corollary}

\begin{proof}
        The algorithm is to run \alg{Amplitude\mbox{-}Estimation}$(U, \Pi, t)$ with $t \leftarrow \ceil{\frac{10\pi}{\sqrt{\varepsilon}}}$ and output {YES} if the estimation result $\widetilde{p} < \frac{\varepsilon}{2}$ or output {NO} otherwise.
        
        It remains to show that if $ p = 0 $ we have $ \widetilde{p} \leq \frac{\varepsilon}{2}$ with high probability, and if $ p > \varepsilon$, we have $\widetilde{p} > \frac{\varepsilon}{2}$ with high probability. If $ p = 0 $, by Theorem \ref{amplitude estimation}, the estimate $\widetilde{p} = 0$ with certainty. 
    If $ p > \varepsilon$, we have the following cases:
    \begin{itemize}
        \item $p < 4\varepsilon$: by Theorem \ref{amplitude estimation}, with probability at least $\frac{8}{\pi^2}$, the estimate $ \widetilde{p} $ satisfies $| \widetilde{p} - p | \leq \frac{2\pi \sqrt{4\varepsilon}}{t} + \frac{\pi^2}{t^2} < \frac{\varepsilon}{2}$, and thus $\widetilde{p} > \frac{\varepsilon}{2}$;

        \item $p \geq 4\varepsilon$: by Theorem \ref{amplitude estimation}, with probability at least $\frac{8}{\pi^2}$, the estimate $\widetilde{p}$ satisfies $| \widetilde{p} - p | \leq \frac{2\pi \sqrt{p}}{t} + \frac{\pi^2}{t^2} < \frac{{p}}{2}$, and thus $\widetilde{p} > \frac{{p}}{2} > \frac{\varepsilon}{2}$.
    \end{itemize}
\end{proof}

\section{Quantum  Complexity of  $\ell^\alpha$-closeness Testing}
\subsection{Quantum Upper Bounds}

In this section, we give a quantum tester (\cref{thm:l2-tester}) for $\ell^2$-closeness testing problem and the $\ell^1$-closeness tester follows by the Cauchy-Schwartz inequality as shown in \cref{corollary2}.

Let $U_p$ and $U_q$ be purified quantum query-access oracles for distributions $p$ and $q$ which work as 
\begin{align*}
    U_p | 0 \rangle_A | 0 \rangle_B = \sum_{i \in [n]} \sqrt{p_i} | \phi_i \rangle_A | i \rangle_B, \\
    U_q | 0 \rangle_A | 0 \rangle_B = \sum_{i \in [n]} \sqrt{q_i} | \psi_i \rangle_A | i \rangle_B.
\end{align*}
The main obstacle in designing closeness testers is that the oracles for two distributions $p$ and $q$ may have different ``directions'' for the same element $i$, i.e., $| \phi_i \rangle \neq | \psi_i \rangle$.  One of the feasible approaches is to block-encode the distribution, e.g., Gily{\'e}n and Li~\cite{gilyen2019distributional}, wherein the probabilities would be encoded as the diagonal elements of some submatrix of the encoding unitary operator. 
However, this approach led to a more complex algorithm.

Define $U_{\mathit{copy}}$ as a unitary operation acting on the Hilbert space $\mathbb{C}^{n} \otimes \mathbb{C}^{n} $ working as $U_{\mathit{copy}}| i \rangle | 0 \rangle = | i \rangle | i \rangle $ for $i = 1, 2, \cdots, n$, which can be constructed by a cascade of CNOT gates.

\begin{lemma}\label{lemma:1}
Given a distribution $p$ over $[n]$ and the corresponding query-access oracle $U_p$, let $|0\rangle_A | 0 \rangle_B | 0 \rangle_C$ be the initial state where the Hilbert space $C$ has the same dimension as the Hilbert space $B$, define unitary $\widetilde{U}_p$ as
\begin{align}
\label{equ:wup}
    \widetilde{U}_p = {\rbra*{U_p^{\dagger}}}_{AB} {\rbra*{{U_{\mathit{copy}}}}}_{BC} \rbra*{U_p}_{AB},
\end{align}
whose circuit is shown in \cref{fig:close1}.
Let $\Pi = | 0 \rangle_A \langle 0 |_A \otimes | 0 \rangle_B \langle 0 |_B \otimes I_C$. Then
    \begin{align}
    \label{equ:up_tidle}
        \widetilde{U}_p |0\rangle_A | 0 \rangle_B | 0 \rangle_C = \sum_{i = 1}^n p_i | 0 \rangle_A | 0 \rangle_B | i \rangle_C + | 0^{\perp} \rangle,
    \end{align}
 and $ \Pi | 0^{\perp} \rangle = 0$.  
\end{lemma}

\begin{figure}[htbp]
    \centering
    \scalebox{0.9}{
    \begin{quantikz}[column sep={1.1cm,between origins}, row sep={0.7cm,between origins}]
        \lstick{$\ket{0}_A$} & \qwbundle{} & \gate[5]{U_p} &  &&& \gate[5]{U_p^\dagger} & \\[0.3cm]
        \lstick[4]{$\ket{0}_B$} &&& \ctrl{4}\gategroup[8,steps=3,style={dashed, inner sep=6pt}]{$U_{\textit{copy}}$}  & & & & \\
        & \setwiretype{q} \wire[l][1]["\scalebox{1.5}{\(\vdots\)}"{below,pos=0.0}]{a} &&&\ctrl{4}& \wire[l][1]["\raisebox{-4ex}{\scalebox{2.0}{\(\cdots\)}}"{below,pos=0.5}]{a} && \\
        \setwiretype{n} &&&&&&& \\
        &&&&& \ctrl{4} && \\
        \lstick[4]{$\ket{0}_C$} & & & \targ{} &&&& \\
        & \setwiretype{q} \wire[l][1]["\scalebox{1.5}{\(\vdots\)}"{below,pos=0.0}]{a} &&& \targ{} &&& \\
        \setwiretype{n}&&&&&&& \\
        &&&&& \targ{} &&
    \end{quantikz}
    }
    \caption{The circuit for constructing the unitary $\widetilde{U}_p$ defined in \cref{equ:wup}. The circuit in the dashed box represents $U_{\mathit{copy}}$.}
    \label{fig:close1}
\end{figure}

{
In the approach of Gily{\'e}n and Li~\cite{gilyen2019distributional}, a key step involves creating a projected unitary encoding with singular values $\sqrt{p_i}$ and then applying the technique of QSVT. Since QSVT is designed for applying a polynomial function to the singular values, the initial step is to determine a (low-degree) polynomial approximation of the target function. The degree of polynomial approximation, which also corresponds to the query complexity of QSVT, typically introduces logarithmic factors to control the approximation error. This can be exemplified by Lemma 11 in~\cite{gilyen2019distributional}, shedding light on why the conclusions of~\cite{gilyen2019distributional} incorporate logarithmic factors.
}

In \cref{lemma:1}, we employ the purified query-access oracle in a completely different way, resulting in more concise algorithms and improved complexities. The unitary $\widetilde{U}_p$ can be interpreted as a block-encoding of a submatrix, with $[p_1, \cdots, p_n]^T$ as its first column, rather than having $\cbra{p_i}$ as singular values. This interpretation can be derived from \cref{equ:up_tidle}, which can be reformulated as:
\begin{align*}
    \rbra*{\bra{0}_A \bra{0}_B \otimes I_C} \widetilde{U}_p \rbra*{|0\rangle_A | 0 \rangle_B | 0 \rangle_C} = \sum_{i = 1}^n p_i | i \rangle_C.
\end{align*}

Although \cref{lemma:1} is similar to Lemma 45 of~\cite{gilyen2019quantum} in form, they are totally different since the unitary in Lemma 45 of~\cite{gilyen2019quantum} is ${\rbra*{U_p^{\dagger}}}_{AB} {\rbra*{{U_{\mathit{swap}}}}}_{BC} \rbra*{U_p}_{AB}$ where ${U_{\mathit{swap}}}$ works as $ {U_{\mathit{swap}}} \ket{i} \ket{j} = \ket{j}\ket{i}$ for $i, j \in \sbra{n}$. 

\begin{proof}[Proof of \cref{lemma:1}]
This can be seen by simple calculations.
For $k \in \{1, \cdots, n\}$, we have
    \begin{align*}
        & \langle 0 |_A \langle 0 |_B \langle k |_C \widetilde{U}_p | 0 \rangle_A | 0 \rangle_B | 0 \rangle_C \\
        ={}  & \langle 0 |_A \langle 0 |_B \langle k |_C {\rbra*{U_p^{\dagger}}}_{AB} {\rbra*{{U_{\mathit{copy}}}}}_{BC} \rbra*{U_p}_{AB} |0\rangle_A | 0 \rangle_B | 0 \rangle_C \\
        ={} &  \rbra*{\sum_{i=1}^n \sqrt{p_i} \langle \phi_i |_A \langle i |_B \langle k |_C } {\rbra*{{U_{\mathit{copy}}}}}_{BC} \rbra*{\sum_{j=1}^n \sqrt{p_j} | \phi_j \rangle_A | j \rangle_B} | 0 \rangle_C \\  
        ={} &  \rbra*{\sum_{i=1}^n \sqrt{p_i} \langle \phi_i |_A \langle i |_B \langle k |_C} \rbra*{\sum_{j=1}^n \sqrt{p_j} | \phi_j \rangle_A | j \rangle_B | j \rangle_C } \\
        ={} &  p_k.
    \end{align*}
\end{proof}

Now we are ready to give the main theorem of this section, i.e., \cref{thm:l2-tester}. The formal description of our (tolerant) $\ell^2$-closeness tester is given in \cref{algorithm1}. All that remains to be done is to analyze the correctness and complexity of the algorithm. 

We will firstly introduce a unitary $U$ that encodes the $\ell^2$-distance between distributions $p$ and $q$. Specifically, the unitary $U$ maps the initial state $| \psi_0 \rangle$ (defined later) to $$\sum_{i=1}^n \frac{p_i - q_i}{2} | 0 \rangle | 0 \rangle | i \rangle | 0 \rangle + | \cdots \rangle,$$ where $| \cdots \rangle$ denotes some unnormalized state that is perpendicular to the first part. The construction of the unitary $U$ can be viewed as a specific application of the technique of linear combination of unitary operators~\cite{LCU, Hamiltonian_Simulation}, which was initially developed for Hamiltonian simulations.

Subsequently, applying the amplitude estimation (see \cref{amplitude estimation}) with high enough accuracy suffices to discriminate whether $p$ and $q$ are close or far enough.

\begin{algorithm}
\caption{(Tolerant) $\ell^2$-closeness tester.}
\label{algorithm1}  
\begin{algorithmic}[1]
    \Require Purified quantum query-access oracles $U_p$ and $U_q$ for distributions $p$ and $q$, $\nu \in (0, 1]$ and $\varepsilon \in (0, 1)$.
    \Ensure \textbf{CLOSE} if $\| p - q \|_2 \leq (1-\nu)\varepsilon $, and \textbf{FAR} otherwise.

    \State Let $U = (H)_D \rbra{\rbra{\widetilde{U}_p}_{ABC} \otimes | 0 \rangle \langle 0 | + \rbra{\widetilde{U}_q}_{ABC} \otimes | 1 \rangle \langle 1 | } (HX)_D$;
    
    \State Let $\Pi = | 0 \rangle_A  \langle 0 |_A\otimes | 0 \rangle_B \langle 0 |_B\otimes I \otimes | 0 \rangle_D \langle 0 |_D$ and $t = 20\pi/\nu\varepsilon$;

    \State Let $\Delta'$ be the result of \alg{Amplitude\mbox{-}Estimation}$(U, \Pi, t)$;

    \If{$\Delta ' < (\frac{1}{4} - \frac{\nu}{8})\varepsilon^2$}
    \State output \textbf{CLOSE};
    \Else
    \State output \textbf{FAR};
    \EndIf

\end{algorithmic}
\end{algorithm}

\begin{theorem}[{\normalfont Restatement of \cref{thm:l2-tester}, Tolerant $\ell^2$-Closeness Tester}]
    For $\nu \in (0, 1]$ and $\varepsilon \in (0, 1)$, there is a quantum tester that can decide whether $\| p - q \|_2 \leq (1-\nu)\varepsilon $ or $\| p - q \|_2 \geq \varepsilon$, with probability at least $\frac{2}{3}$, using $O(\frac{1}{\nu\varepsilon})$ queries to the oracles of $p$ and $q$.
\end{theorem}

\begin{figure}[htbp]
    \centering
    \scalebox{0.9}{
    \begin{quantikz}[column sep={1.1cm,between origins}, row sep={0.7cm,between origins}]
        \lstick{$\ket{0}$} & \gate{X} & \gate{H} & \ctrl{1} & \octrl{1} & \gate{H} & \\
        \lstick{$\ket{0}_{ABC}$} & \qwbundle{} & & \gate{\widetilde{U}_p} & \gate{\widetilde{U}_q} && 
    \end{quantikz}}
    \caption{The circuit for constructing the unitary $U$ defined in \cref{equ:u}.}
    \label{fig:close2}
\end{figure}
    
\begin{proof}
    Let the initial state be
    \begin{align*}
        | \psi_0 \rangle = | 0 \rangle_A | 0 \rangle_B | 0 \rangle_{C} | 0 \rangle_D,
    \end{align*}
     where the system $C$ has the same dimension as the system $B$, and the last system is a qubit (a two-dimensional Hilbert space).

    Define $\widetilde{U}_p = (U_p^{\dagger} \otimes I ) (I \otimes U_{\mathit{copy}}) (U_p \otimes I)$ (similar for $\widetilde{U}_q$) as in \cref{lemma:1}. Define 
    \begin{equation}
    \label{equ:u}
    \begin{aligned}
         U = (H)_D \left(\rbra*{\widetilde{U}_p}_{ABC} \otimes | 0 \rangle \langle 0 | +\rbra*{\widetilde{U}_q}_{ABC} \otimes | 1 \rangle \langle 1 | \right) (HX)_D,
    \end{aligned}
    \end{equation}
    where $H$ is the Hadamard gate, and $X$ is the Pauli-X gate. The circuit for constructing the unitary $U$ is shown in \cref{fig:close2}.
    
    We have
    \begin{align*}
        | \psi \rangle
        & = U | \psi_0 \rangle \\
        & = (H)_D \rbra*{\rbra*{\widetilde{U}_p}_{ABC} \otimes | 0 \rangle \langle 0 | + \rbra*{\widetilde{U}_q}_{ABC} \otimes | 1 \rangle \langle 1 | } (HX)_D | 0 \rangle_A | 0 \rangle_B | 0 \rangle_C | 0 \rangle_D \\
        & = (H)_D \rbra*{\rbra*{\widetilde{U}_p}_{ABC} \otimes | 0 \rangle \langle 0 | + \rbra*{\widetilde{U}_q}_{ABC} \otimes | 1 \rangle \langle 1 | } | 0 \rangle_A | 0 \rangle_B | 0 \rangle_C \rbra*{\frac{1}{\sqrt{2}} | 0 \rangle - \frac{1}{\sqrt{2}} | 1 \rangle}_D \\
        & = (H)_D \rbra*{\sum_{i=1}^n \frac{p_i}{\sqrt{2}} | 0 \rangle_A | 0 \rangle_B | i \rangle_C | 0 \rangle_D - \sum_{i=1}^n \frac{q_i}{\sqrt{2}} | 0 \rangle_A | 0 \rangle_B | i \rangle_C | 1 \rangle_D + | 0^{\perp} \rangle_{ABCD}} \\
        & = \sum_{i=1}^n \frac{p_i - q_i}{2} | 0 \rangle_A | 0 \rangle_B | i \rangle_C | 0 \rangle_D + \sum_{i=1}^n \frac{p_i + q_i}{2} | 0 \rangle_A | 0 \rangle_B | i \rangle_C | 1 \rangle_D + | 0^{\perp} \rangle_{ABCD},
    \end{align*}
    where $( | 0 \rangle \langle 0 | \otimes | 0 \rangle \langle 0 | \otimes I \otimes I) | 0^{\perp} \rangle = 0$, and the third equation is from Lemma \ref{lemma:1}.

    Define an orthogonal projector $\Pi$ as 
    \begin{align*}
        \Pi = | 0 \rangle_A  \langle 0 |_A\otimes | 0 \rangle_B \langle 0 |_B\otimes I \otimes | 0 \rangle_D \langle 0 |_D.
    \end{align*}
    
    Then, we have 
    \begin{align}
        \label{equ:pipsi}
        \| \Pi \psi \rangle \|^2 = \frac{1}{4} \sum_{i = 1}^n (p_i - q_i)^2 = \frac{\| p - q \|_2^2 }{4}.
    \end{align} 
    Let $\Delta = \frac{\| p - q \|_2^2 }{4} $, and $\Delta '$ be the estimate for $\Delta$ computed by \alg{Amplitude\mbox{-}Estimation}$(U, \Pi, t)$ with $t = \frac{20\pi}{\nu \varepsilon}$ as in \cref{algorithm1}.

    It remains to show that if $\| p - q \|_2 \leq (1-\nu) \varepsilon$ we have $\Delta' \leq \rbra*{\frac{1}{4} - \frac{\nu}{8}}\varepsilon^2$ with high probability, and if $\| p - q \|_2 \geq \varepsilon$, we have $\Delta' > \rbra*{\frac{1}{4} - \frac{\nu}{8}}\varepsilon^2$ with high probability. If $\| p - q \|_2 \leq (1-\nu) \varepsilon$, then $\Delta \leq \frac{\varepsilon^2 - \nu \varepsilon^2 }{4}$; by Theorem \ref{amplitude estimation}, with probability at least $\frac{8}{\pi^2}$, the estimate $\Delta '$ satisfies $| \Delta ' - \Delta| \leq \frac{2\pi \sqrt{\Delta}}{t} + \frac{\pi^2}{t^2} < \frac{\nu\varepsilon^2}{8}$, and thus $\Delta' \leq \rbra*{\frac{1}{4} - \frac{\nu}{8}}\varepsilon^2$. 
    If $\| p - q \|_2 \geq \varepsilon$, then $\Delta \geq \frac{ \varepsilon^2 }{4}$. We have the following cases:
    \begin{itemize}
        \item $\Delta \leq \varepsilon^2$: by Theorem \ref{amplitude estimation}, with probability at least $\frac{8}{\pi^2}$, the estimate $\Delta ' $ satisfies $| \Delta ' - \Delta | \leq \frac{2\pi \varepsilon}{t} + \frac{\pi^2}{t^2} < \frac{\nu\varepsilon^2}{8}$, and thus $\Delta' > \rbra*{\frac{1}{4} - \frac{\nu}{8}}\varepsilon^2$;

        \item $\Delta > \varepsilon^2$: by Theorem \ref{amplitude estimation}, with probability at least $\frac{8}{\pi^2}$, the estimate $\Delta ' $ satisfies $| \Delta ' - \Delta | \leq \frac{2\pi \sqrt{\Delta}}{t} + \frac{\pi^2}{t^2} < \frac{\Delta}{2}$, and thus $\Delta ' > \frac{\varepsilon^2}{2} > \rbra*{\frac{1}{4} - \frac{\nu}{8}}\varepsilon^2$. 
    \end{itemize}
    
    \textbf{Complexity.} The query complexity of constructing $U$ is $O(1)$, and it takes $O\rbra*{\frac{1}{\nu \varepsilon}}$ queries to $U$ to run amplitude estimation. Therefore, the overall query complexity of Algorithm \ref{algorithm1} is $O\rbra*{\frac{1}{\nu \varepsilon}}$.
\end{proof}

{
The proof of \cref{thm-k-uniform} actually gives an $\ell^2$-distance estimator. 

\begin{lemma}
    Suppose that $p$ and $q$ are probability distributions over a finite sample space. For $\varepsilon \in (0, 1)$, there is a quantum estimator that estimates $\Abs{p - q}_2$ to within additive error $\varepsilon$ with probability at least $\frac{2}{3}$, using $O(\frac{1}{\varepsilon^2})$ queries to the purified quantum query-access oracles of $p$ and $q$.
\end{lemma}
\begin{proof}
    As seen from \cref{equ:pipsi}, the unitary $U$ defined in \cref{equ:u} encodes the (squared) $\ell^2$ distance between $p$ and $q$, which can be estimated to within additive error $\varepsilon$ by amplitude estimation, using  $O(\frac{1}{\varepsilon^2})$ queries to the oracles of $p$ and $q$.
\end{proof}
}

\subsection{Quantum Lower Bounds}

It was shown in \cref{corollary1} and \cref{corollary2} that the quantum query complexity for the $\ell^2$-closeness and $\ell^1$-closeness testing problem is $O\rbra{\frac{1}{\varepsilon}}$ and $O\rbra{\frac{\sqrt{n}}{\varepsilon}}$, respectively. In the following, we show that both testers achieve optimal dependence on $\varepsilon$. This is done by a reduction from the problem of distinguishing probability distributions.

Belovs~\cite{belovs2019quantum} considered a related problem: given two fixed probability distributions $p$ and $q$, and given a purified quantum query-access oracle for one of them,
the task is to determine which one, $p$ or $q$, the oracle encodes. He showed that the quantum query complexity of this distinguishing problem is $\Theta\rbra{1/d_{\text{H}}\rbra{p, q}}$, achieving a quadratic speedup compared to classical approaches, where $d_{\textup{H}}\rbra{p, q}$ is the Hellinger distance defined by
\begin{align*}
    d_{\textup{H}}\rbra{p, q} = \sqrt{\frac 1 2 \sum_{i \in \sbra{n}} \rbra*{\sqrt{p_i} - \sqrt{q_i}}^2}.
\end{align*}

\begin{theorem} [{\cite[Theorem 4]{belovs2019quantum}}]
\label{thm:hellinger}
    Suppose that $p$ and $q$ are probability distributions over $\sbra{n}$, and given a purified quantum query-access oracle for one of them, then the quantum query complexity for distinguishing between two probability distributions $p$ and $q$ is $\Omega\rbra{1/d_{\text{H}}\rbra{p, q}}$.
\end{theorem}

\begin{theorem}[{\normalfont Restatement of \cref{lb:closeness}}]
    For $\alpha = \cbra{1, 2}$, the quantum query complexity of $\ell^\alpha$-closeness testing  is $\Omega(\frac{1}{\varepsilon})$.
\end{theorem}

\begin{proof}
    We reduce the distribution discrimination problem to the closeness testing problem. Let $n \geq 2$ be an even integer and $\varepsilon \in (0, 1)$. For $\alpha = 2$,  consider the case that $p_1 = p_2 = \frac{1}{2}, p_i = 0$ for $i > 2$ and $q_1 = \frac{1 - \varepsilon}{2}, q_2 = \frac{1 + \varepsilon}{2}, q_i = 0$ for $i > 2$.
    Then, the $\ell^2$-distance between $p$ and $q$ is 
    \begin{align*}
        \Abs{p-q}_2 & = \frac{\varepsilon}{\sqrt{2}}.
    \end{align*}
    And the Hellinger distance between $p$ and $q$ is
    \begin{align*}
        d_{\text{H}}\rbra{p, q} 
        & = \sqrt{\frac{1}{2}\sbra*{\rbra*{\sqrt{\frac{1-\varepsilon}{2}} - \sqrt{\frac{1}{2}}}^2 + \rbra*{\sqrt{\frac{1+\varepsilon}{2}} - \sqrt{\frac{1}{2}}}^2}} \\
        & = \sqrt{1 - \frac{\sqrt{1+\varepsilon}}{2} - \frac{\sqrt{1-\varepsilon}}{2}} \\
        & \leq \varepsilon.
    \end{align*} 
    where we have use the inequality $ 1 - \frac{\sqrt{1+\varepsilon}}{2} - \frac{\sqrt{1-\varepsilon}}{2} \leq \varepsilon^2 $ in the last inequality.
    By \cref{thm:hellinger}, we know that distinguishing between $p$ and $q$ requires $\Omega\rbra{1/\varepsilon}$ queries.     
    For $\alpha = 1$, consider the case that $p_i = \frac{1}{n}$ and $q_i = \frac{1 + (-1)^{i}\varepsilon}{n}$ for every $i \in [n]$. A similar calculation shows that
    \begin{align*}
        \Abs{p-q}_1 & = \varepsilon,
    \end{align*}
    and 
    \begin{align*}
        d_{\text{H}}\rbra{p, q} 
        & = \sqrt{\frac{1}{2}\sbra*{\frac{n}{2}\rbra*{\sqrt{\frac{1-\varepsilon}{n}} - \sqrt{\frac{1}{n}}}^2 + \frac{n}{2}\rbra*{\sqrt{\frac{1+\varepsilon}{n}} - \sqrt{\frac{1}{n}}}^2}} \\
        & = \sqrt{1 - \frac{\sqrt{1+\varepsilon}}{2} - \frac{\sqrt{1-\varepsilon}}{2}} \\
        & \leq \varepsilon.
    \end{align*} 
    A similar argument also yields an $\Omega\rbra{1/\varepsilon}$ lower bound. 
\end{proof}

\paragraph{Discussion}
In the classical setting, the sample complexity of $\ell^2$-closeness testing problem is $\Theta\rbra{\frac{\max\cbra{\| p \|_2, \| q \|_2}}{\varepsilon^2}}$~\cite{chan2014optimal}, where the dependence on $\| p \|_2$ and $\| q \|_2$ is important, since this refined complexity may be much lower than the trivial upper bound $O\rbra{\frac{1}{\varepsilon^2}}$ that is derived by simply setting $\max\cbra*{\| p \|_2, \| q \|_2} = O(1)$. 
By comparison, the query complexity (i.e., $\Theta\rbra{\frac{1}{\varepsilon}}$) of our quantum tester does not depend on $\| p \|_2$ or $\| q \|_2$. If we can introduce a dependency on them as the classical one, our quantum $\ell^2$-closeness tester may be further improved. 
And in certain cases (i.e., $\max\cbra{\| p \|_2, \| q \|_2} = \Omega(1)$), our quantum tester achieves a quadratic speedup over the best classical one.

\section{Quantum  Complexity of $k$-Wise Uniformity Testing}
\subsection{Quantum Upper Bound}

In this section, we explore the potential quantum speedup in testing whether a given distribution is $k$-wise uniform or far from being $k$-wise uniform under the metric of total variation distance. Recall that the sample space of the probability distribution considered here is $\{0, 1\}^n$. 

Let's first briefly introduce some concepts and notations. 
As in~\cite{o2018closeness}, for a distribution $p$, define its density function $\varphi\colon \{0, 1\}^n \rightarrow \mathbb{R}^{\geq 0}$ as 
\begin{align*}
    \varphi(x) = 2^np_x,\quad \forall {x} \in \{0, 1\}^n.
\end{align*}
In the following, we identify the distribution $p$ with its density function $\varphi$. Let ${x} \sim \varphi$ denote that ${x}$ is a random variable drawn from the associated distribution with density function $\varphi$ and let ${x \sim \{0, 1\}^n}$ denote that $x$ is drawn from $\{0, 1\}^n$ uniformly.

Next, we briefly introduce the basic concepts of Fourier analysis of Boolean functions~\cite{o2014analysis}. Given any function $f \colon \{0, 1\}^n \rightarrow \mathbb{R}$, for any index subset $S \subseteq [n]$, the Fourier coefficient of $f$ on $S$, denoted by $\hat{f}(S)$, is defined as follows:
\begin{align*}
    \hat{f}(S) = \mathop{\mathbb{E}}_{x \sim \{0, 1\}^n}[f(x)\chi_S(x)],
\end{align*}
where 
\begin{align*}
    \chi_S(x) = (-1)^{\sum_{i \in S}x_i}.
\end{align*}
Also, $f$ has a unique representation as a multilinear polynomial:
\begin{align*}
    f(x) = \sum_{S \subseteq [n]} \hat{f}(S) \chi_S(x).
\end{align*}

The following lemmas relate the property of $k$-wise uniformity to the Fourier coefficients of the density function. The first lemma is well-known, and the second lemma, due to O'Donnell and Zhao~\cite{o2018closeness}, improves the result from  Alon~\textit{et~al.}~\cite{alon2007testing} by a factor of $\rm{polylog}(n)$.

\begin{lemma}[\cite{alon2007testing, o2018closeness}]
\label{lemma:k-wise1}
    Given a probability distribution with density function $\varphi$ over $\{0, 1\}^n$, if $\varphi$ is a $k$-wise uniform distribution, then we have  $\hat{\varphi}(S) = 0$ for all $S \subseteq [n] $ satisfying $1 \leq |S| \leq k$.
\end{lemma}

\begin{lemma}[{\cite[Theorem 1.1]{o2018closeness}}]
\label{lemma:k-wise2}
    Given a probability distribution with density function $\varphi$ over $\{0, 1\}^n$, if $\varphi$ is $\varepsilon$-far from any $k$-wise uniform distribution, then we have 
    \begin{align*}
    \sqrt{\sum_{S \subseteq [n]\colon 1 \leq |S| \leq k} \hat{\varphi}(S)^2} > \frac {\varepsilon }{e^k}.
 \end{align*} 
\end{lemma}

By the above two lemmas, We can see that the quantity ${\sum_{S \subseteq [n]\colon 1 \leq |S| \leq k} \hat{\varphi}(S)^2}$ indicates whether $\varphi$  is $k$-wise uniform  or not. In the following lemma, we give a unitary operator $U$ encoding the quantity in the amplitude of some quantum state.
\begin{lemma}\label{lemma:alg2}
    Given a distribution $p$ with density function $\varphi$ over $\{0, 1\}^n$ and the corresponding oracle $U_p$, let the initial state be
    $
        \label{equ:init-state-k-wise}
        | \psi_1 \rangle = | 0 \rangle ^{\otimes n} |0 \rangle | 0 \rangle ^{\otimes n}
    $
    where the first register consists of $n$ qubits, the second and third registers form the workspace of $U_p$ and define $U$ as \cref{alg:encode-k-wise}, then
    \begin{align*}
        U |0 \rangle^{\otimes n}& | 0 \rangle \ket{0}^{\otimes n} = \frac{1}{\sqrt{M^n_k}} \sum_{S \subseteq [n]\colon 1 \leq |S| \leq k} \hat{\varphi}(S) | S  \rangle | 0 \rangle | 0 \rangle ^{\otimes n} + | 0^\perp \rangle,
    \end{align*}
    where $| 0^\perp \rangle$ is some unnormalized state such that $ \Pi | 0^\perp \rangle = 0$ with $ \Pi = I \otimes | 0 \rangle \langle 0 | \otimes | 0 \rangle\langle 0 |^{\otimes n} $. 
\end{lemma}

\begin{algorithm}
    \caption{Definition of the operator $U$.}
    \label{alg:encode-k-wise}  
    \begin{algorithmic}[1]
    \State Perform $V_k \otimes U_p$, where $V_k$ is a unitary operation working as $$V_k | 0 \rangle^{\otimes n} = \frac{1}{\sqrt{M^n_k}} \sum_{S \subseteq [n]\colon 1 \leq |S| \leq k} | {S} \rangle;$$ \label{step1} 

    \For{$i = 1 \to n$} \label{step2}
    \State With the $i^{\mathit{th}}$ qubit of the first register as the control qubit and the $i^{\mathit{th}}$ qubit of the third register as the target qubit, perform the controlled-$Z$ gate; \label{step3}
    \EndFor

    \State Perform $(I \otimes U_p)^\dagger$; \label{step4}
    \end{algorithmic}
\end{algorithm}

\begin{figure}[htbp]
    \centering
    \scalebox{0.9}{
    \begin{quantikz}[column sep={1.1cm,between origins}, row sep={0.8cm,between origins}]
        \lstick[4]{$\ket{0}^{\otimes n}$} & & \gate[4]{V_k}& \ctrl{5} &&&& \\
        & \wire[l][1]["\scalebox{1.5}{\(\vdots\)}"{below,pos=0.0}]{a} &&& \ctrl{5}\wire[r][1]["\raisebox{-4ex}{\scalebox{2.0}{\(\cdots\)}}"{below,pos=0.5}]{a} &&& \\
        \setwiretype{n}&&&& &&& \\
        &&&&& \ctrl{5} && \\
        \lstick{$\ket{0}$} & \qwbundle{} & \gate[5]{U_p} &&&& \gate[5]{U_p^{\dagger}} & \\
        \lstick[4]{$\ket{0}^{\otimes n}$} &&& \gate{Z} &&&&\\
        & \wire[l][1]["\scalebox{1.5}{\(\vdots\)}"{below,pos=0.0}]{a} && & \gate{Z} &&&\\
        \setwiretype{n}&&&& &&& \\
        &&&&& \gate{Z} &&\\
    \end{quantikz}}
    \caption{The circuit for constructing the unitary $U$ defined in \cref{alg:encode-k-wise}.}
    \label{fig:k-wise}
\end{figure}    

{
The ciruit for \cref{alg:encode-k-wise} is shown in \cref{fig:k-wise}. In \cref{alg:encode-k-wise}, the unitary operator $V_k$ is used to prepare the state
\begin{align*}
    \frac{1}{\sqrt{M^n_k}} \sum_{S \subseteq [n]\colon 1 \leq |S| \leq k} | {S} \rangle,
\end{align*}
where $| {S} \rangle$ denotes an $n$-qubit state such that $| {S} _i \rangle = | 1 \rangle$ if and only if $i \in S$ for all $i \in [n]$. This is a special symmetric pure state, i.e., a superposition state of Dicke states~\cite{Bastin_2009}. A state is symmetric if it is invariant under permutation of the qubits. In the context of the universal gate set comprising CNOT and arbitrary single-qubit gates,  B{\"a}rtschi and Eidenbenz~\cite{bartschi2019deterministic} demonstrated that every symmetric pure $n$-qubit state can be efficiently prepared using a circuit of size $O(n^2)$ and depth $O(n)$. In particular, CNOT gates and arbitrary single-qubit $R_y(\theta)$ gates are sufficient for implementing $V_k$. Importantly, these upper bounds remain valid, even on spin chain architectures, where the qubits are laid out in a line,
and all CNOT gates must act only on adjacent (nearest-neighbor) qubits. This is a more realistic assumption for most Noisy Intermediate Scale Quantum (NISQ)~\cite{preskill2018quantum} devices than the standard all-to-all connectivity.
}

\begin{proof}[Proof of \cref{lemma:alg2}]
    Let's show how $U$ works step by step.
    After performing the operator $(V_k \otimes U_p)$, we have 
    \begin{align*}
        | \psi_2 \rangle 
        & = \rbra*{V_k \otimes U_p} | \psi_1 \rangle \\
        & = \rbra*{\frac{1}{\sqrt{M^n_k}} \sum_{\substack{S \subseteq [n]\colon\\ 1 \leq |S| \leq k}} |{S} \rangle} \rbra*{\sum_{x \in \{0, 1\}^n} \sqrt{p_x} | \phi_x \rangle | x \rangle}.
    \end{align*}

    After the for-iteration, we have 
    \begin{align*}
        | \psi_3 \rangle 
        = \frac{1}{\sqrt{M^n_k}} \sum_{S \subseteq [n]\colon 1 \leq |S| \leq k} \left( \sum_{\substack{x \in \{0, 1\}^n\colon \\ \chi_S(x) = 1}} \sqrt{p_x} | {S}  \rangle | \phi_x \rangle | x \rangle
        -  \sum_{\substack{x \in \{0, 1\}^n\colon \\ \chi_S(x) = -1}} \sqrt{p_x} | {S}  \rangle | \phi_x \rangle | x \rangle \right).
    \end{align*}

    Let $p_S = \Pr_{x \sim \{0, 1\}^n} [\chi_S(x) = -1]$, we have 
    \begin{align*}
        \hat{\varphi}(S) 
        & = \mathop{\mathbb{E}}_{x \sim \{0, 1\}^n} [\varphi(x)\chi_S(x)] \\
        & = \Pr_{x \sim \varphi} [\chi_S(x) = 1] - \Pr_{x \sim \varphi} [\chi_S(x) = -1] \\
        & = 1 - 2\Pr_{x \sim \varphi} [\chi_S(x) = -1] \\
        & = 1 - 2 p_S.
    \end{align*}

    To compute $| \psi_4\rangle = (I \otimes U_p)^\dagger \ket{\psi_3}$, for $\widetilde{S}\colon 1 \leq | \widetilde{S} | \leq k$, we have
    \begin{align*}
        \langle \widetilde{S}  | \langle 0 | \langle 0 |^{\otimes n} |\psi_4 \rangle 
        & = \langle \widetilde{S}|  \langle 0  | \langle 0 |^{\otimes n} \rbra*{I \otimes U_p}^\dagger| \psi_3 \rangle \\
        & = \rbra*{\sum_{x\in \cbra{0, 1}^n}\sqrt{p_x}\bra{\widetilde{S}}\bra{\phi_x}\bra{x}} | \psi_3 \rangle \\
        & = \frac{1 - 2p_{\widetilde{S}}}{\sqrt{M^n_k}} = \frac{\hat{\varphi}\rbra{{\widetilde{S}}}}{\sqrt{M^n_k}}.
    \end{align*}
    It's not difficult to see that for $\widetilde{S} \subseteq [n]$ such that $|{\widetilde{S}}| = 0$ or $| {\widetilde{S}} | > k $, we have $\langle {\widetilde{S}} | \langle 0 | \langle 0 |^{\otimes n} | \psi_4 \rangle = 0$.
  Thus we have 
    \begin{align*}
        | \psi_4 \rangle 
        = \frac{1}{\sqrt{M^n_k}} \sum_{S \subseteq [n] \colon 1 \leq |S| \leq k} \hat{\varphi}(S) | S  \rangle | 0 \rangle | 0 \rangle ^{\otimes n} + | 0^\perp \rangle,
    \end{align*}
    where $| 0^\perp \rangle$ is some unnormalized state such that $ \Pi | 0^\perp \rangle = 0$ with $ \Pi = I \otimes | 0 \rangle \langle 0 | \otimes | 0 \rangle\langle 0 |^{\otimes n} $. 
\end{proof}

\begin{algorithm}[htb]
\caption{Quantum $k$-wise uniformity tester.}
\label{algorithm2}  
\begin{algorithmic}[1]
    \Require {Purified quantum query-access oracle $U_p$ for the distribution $p$ over $\cbra{0, 1}^n$; $\varepsilon \in (0, 1)$.}
    \Ensure {\textbf{YES} if $p$ is $k$-wise uniform, or \textbf{NO} if $p$ is $\varepsilon$-far from any $k$-wise uniform distribution.}

    \State Prepare the initial state $| \psi_1 \rangle = | 0 \rangle ^{\otimes n} | 0 \rangle | 0 \rangle ^{\otimes n} $;

    \State Let $U$ defined as \cref{alg:encode-k-wise};

    \State Let $ \Pi = I \otimes | 0 \rangle \langle 0 | \otimes | 0 \rangle\langle 0 |^{\otimes n} $;

    \State Output the result of \alg{Zero\mbox{-}Tester}$(U, \Pi, \frac{\varepsilon^2}{e^{2k}M^n_k})$ (\cref{corollary3});
\end{algorithmic}
\end{algorithm}

Now, we are going to prove another main result of this article, that is, Theorem \ref{thm-k-uniform}.

\begin{theorem}[{\normalfont Restatement of \cref{thm-k-uniform}}]
    Let $k \geq 1$ be a constant, there is a quantum tester for $k$-wise uniformity testing with query complexity $O\rbra{\frac{\sqrt{n^k}}{\varepsilon}}$. 
\end{theorem}

\begin{proof}
    We propose \cref{algorithm2} for $k$-wise uniformity testing. Note that $$\| \Pi U | \psi_1 \rangle \| ^2 = \frac{1}{M^n_k}\sum_{S \subseteq [n] \colon 1 \leq |S| \leq k} \hat{\varphi}(S)^2, $$ which is exactly the quantity we want to estimate. Again let $\Delta = \| \Pi U | \psi_1 \rangle \| ^2$. By \cref{lemma:k-wise1}, if $p$ is $k$-wise uniform, we have $\Delta = 0$. By \cref{lemma:k-wise2}, if $p$ is $\varepsilon$-far from any $k$-wise uniform distribution, we have $\Delta > \frac{\varepsilon^2}{e^{2k}M^n_k}$.  Therefore, running  \alg{Zero\mbox{-}Tester} in \cref{corollary3} with $\varepsilon \leftarrow \frac{\varepsilon^2}{e^{2k}M^n_k}$ is enough to determine whether $p$ is $k$-wise uniform or $\varepsilon$-far from any $k$-wise uniform distribution. The query complexity is $O\rbra{\frac{\sqrt{n^k}}{\varepsilon}}$ by noting that $M^n_k \leq n^k$, which completes the proof.
\end{proof}

We would like to point out an interesting difference between our quantum algorithm and the best classical one~\cite{o2018closeness}. It can be seen from the above proof that the analysis of our quantum algorithm is intuitive and concise, whereas the classical analysis is much more involved. In \cite{o2018closeness}, a proper random variable whose expectation corresponds to the quantity we want to estimate is constructed and in order to apply Chebyshev's inequality, estimating the variance of random variables is inevitable, which may require a lot of computation.

\subsection{Quantum Lower Bound}
In this section, we show a lower bound for the $k$-wise uniform testing problem. We need the following lemma from  Alon~\textit{et~al.}~\cite{alon2007testing}, which states that if $p$ is a uniform distribution over a  randomly sampled subset of $\{0, 1\}^n$, up to size $\frac{M^n_k}{n\varepsilon^2}$, then $p$ is $O(\varepsilon)$-far from any $k$-wise uniform distribution with high probability.

Let $\mathcal{Q} = \cbra{a_1, a_2, \dots, a_m}$ be a multiset of cardinality $m$, the uniform distribution over $\mathcal{Q}$, denoted by $\mathop{\mathit{Unif}}_{{\mathcal{Q}}}$, is defined as $$\Pr_{x \sim \mathop{\mathit{Unif}}_{{\mathcal{Q}}}}\sbra{x = y} = \frac{\abs{\cbra{i \in \sbra{m}\colon a_i = y}}}{m}$$ for every $y \in \mathcal{Q}$.

\begin{lemma} [Random Distribution Lemma, {\cite[Lemma 3.6]{alon2007testing}}]\label{lemma:random-subset}
    Let $k > 2$ be a constant and $Q = \frac{M^n_k}{n\varepsilon^2} < 2^{n^{1/3}}$, with $0 < \varepsilon \leq 1$. Let $\mathcal{Q}$ be a random multiset consisted of ${Q}$ strings that are sampled uniformly at random from $\{0, 1\}^n$, then distribution $\mathop{\mathit{Unif}}_{{\mathcal{Q}}}$ is at least $0.228\varepsilon$-far from any $k$-wise uniform distribution, with probability $99\%$ for sufficiently large $n$ and sufficient small $\varepsilon$.
\end{lemma}

We will reduce the collision problem to the $k$-wise uniformity testing problem. Given a function $f\colon [n] \rightarrow [n]$, the collision problem requires distinguishing whether $f$ is one-to-one or $r$-to-one, under the promise that one of these two cases holds. The lower bound for the collision problem is shown as follows:

\begin{theorem}[Collision Problem~\cite{aaronson2004quantum, ambainis2005polynomial, kutin2005quantum}]\label{the:collision}
    Let $n > 0$ and $r \geq 2$ be integers such that $r \mid n$, and let $f\colon [n] \rightarrow [n]$ be a function given as an oracle $O_f\colon |i\rangle | b \rangle \rightarrow | i \rangle | b + f(i) \mod n \rangle$, and promise that it is either one-to-one or $r$-to-one. Then any quantum algorithm for distinguishing these two cases must make $\Omega((n/r)^{1/3}))$ queries to $O_f$.
\end{theorem}

In order to complete the reduction, two slight modifications to the above theorem are necessary. Firstly, we modify the oracle formulation to $O_f\colon |i\rangle | b \rangle \rightarrow | i \rangle | b \oplus f(i) \rangle$ where $\oplus$ denotes bitwise XOR. Secondly, we allow the quantum algorithm to make incorrect decisions on a negligible fraction of  each kind of functions. Importantly, the lower bound of $\Omega((n/r)^{1/3}))$ remains valid. The formal statement is provided in the following corollary.

\begin{corollary}
\label{cor:collision}
    Let $n > 0$ and $r \geq 2$ be integers such that $r \mid n$, and $f \colon [n] \rightarrow [n]$ be a function that is either one-to-one or $r$-to-one. Assume that the inputs and outputs of the functions are encoded using $\lceil{\log_2 n}\rceil$ qubits, and the function $f$ is given as an oracle $O_f\colon |i\rangle | b \rangle \rightarrow | i \rangle | b \oplus f(i) \rangle$. 
    Then any quantum algorithm that determines (with probability at least $2/3$) the type of the input function $f$ for $99\%$ one-to-one functions and $99\%$ $r$-to-one functions must make $\Omega((n/r)^{1/3}))$ queries to $O_f$.
\end{corollary}

The proof is almost the same as the one in~\cite{kutin2005quantum}, and for completeness, we have included it in the appendix. 
Now we are ready to prove the lower bound for the $k$-wise uniformity testing problem.
{

\begin{theorem}[{\normalfont Restatement of \cref{lb:k-wise}}]
    Let $k > 2$ be a constant, $n$ be a sufficiently large integer,  and $\varepsilon \in (0, 0.228)$ be a sufficiently small  real number with $0.228^2\frac{M^n_k}{n\varepsilon^2} < 2^{n^{1/3}}$. Then the quantum query complexity of $k$-wise uniformity testing is $\Omega\rbra{\frac{n^{(k-1)/3}}{\varepsilon^{2/3}}}$.
\end{theorem}

\begin{proof}
    Let $Q$ be the maximum integer such that $Q \leq 0.228^2 \frac{M^n_k}{n\varepsilon^2}$ and $Q \mid 2^n$. 
    Let $\mathcal{Q}$ be a random multiset consisting of ${Q}$ strings that are sampled uniformly at random from $\{0, 1\}^n$. By the Birthday Paradox, the probability that all the strings in $\mathcal{Q}$ are distinct is 
    \begin{align*}
        p & = \rbra*{1 - \frac{0}{2^n}} \cdot \rbra*{1 - \frac{1}{2^n}} \cdots \rbra*{1 - \frac{Q-1}{2^n}} \\
        & \geq 1 - \frac{Q(Q-1)}{2^{n+1}},
    \end{align*}
    which is $1 - o\rbra{1}$ when $Q \leq 0.228^2 \frac{M^n_k}{n\varepsilon^2} \ll 2^{n/2}$.  From now on, let's assume this is the case. 
    
    Consider the following distribution on $f\colon \cbra{0, 1}^n \rightarrow \cbra{0, 1}^n$: with probability $1/2$, $f$ is a random $1$-to-$1$ function, and with probability $1/2$, $f$ is a random $2^n/Q$-to-$1$ function with support on $\mathcal{Q}$. Then by Theorem \ref{the:collision} any quantum algorithm distinguishing these two cases must make $\Omega(Q^{1/3})$ queries. On the other hand, in the first case the distribution defined by $f$ is the uniform distribution, which is $k$-wise uniform for all $1 \leq k \leq n$. In the second case, by Lemma \ref{lemma:random-subset}, we have the distribution defined by $f$ is at least $\varepsilon$-far from any $k$-wise uniform distribution with probability $99\%$ for sufficiently large $n$ and sufficiently small $\varepsilon$. 
    Therefore, a quantum algorithm for $k$-wise uniformity testing problem can disdinguish all random $1$-to-$1$ functions from almost all random $2^n/Q$-to-$1$ functions, therefore inheriting the quantum query complexity $\Omega\rbra{Q^{1/3}}$ for the collision problem (see \cref{the:collision} and \cref{cor:collision}).
    
    To completing the reduction, it remains to construct the (controlled, inverse) purified quantum query-access oracle (see \cref{purified quantum query-access}) using constant queries to the oracle $O_f$ for the collision problem (see \cref{cor:collision}). The oracle $O_f$ could be viewed as a discrete quantum query-access model defined in \cref{def:dqqa}. The circuit for the purified quantum query-access oracle could be constructed as follows: prepare the initials state $\ket{0}^{\otimes n}\ket{0}^{\otimes n}$, apply $H^{\otimes n}\otimes I$ to create a uniform superposition state of all $2^n$ different inputs and then query $O_f$. Since $O_fO_f = I$, the inverse quantum query-access oracle could be $O_f\rbra{H^{\otimes n}\otimes I}$. In order to build the controlled purified quantum query-access oracle, it remains to show how to implement the controlled version of $O_f$, which could be done as follows: append a new auxiliary register, query $O_f$ and store the result in the auxiliary register rather than the original output register, and then copy the content of the auxiliary register to the output register under the control of control bit(s), finally uncompute the content in the auxiliary register by querying $O_f$ again.
    
    As mentioned above, we have reduced the collision problem to the $k$-wise uniform testing problem. Therefore, the quantum algorithm for $k$-wise uniform testing problem inherits the quantum query complexity $\Omega\rbra{Q^{1/3}}$ for the collision problem, which is $\Omega\rbra{Q^{1/3}} = \Omega\rbra{\frac{1}{\varepsilon^{2/3}}(\frac{n}{k})^{\frac{k-1}{3}}}$.
\end{proof}

}

\section{Conclusion and Open Problems}

In this article, we have proposed concise and efficient quantum algorithms for testing whether two classical distributions are close or far  enough under the metrics of $\ell^1$-distance and $\ell^2$-distance, with  query complexity  $O(\sqrt{n}/\varepsilon)$ and $O(1/\varepsilon)$, respectively, outperforming the prior best algorithms. In particular, both of our testers achieve optimal dependence on $\varepsilon$. 
In addition, we also obtained the first quantum algorithm for $k$-wise uniformity testing,  
achieving a quadratic speedup over the state-of-the-art classical algorithm.
We also provided a non-matching lower bound for this problem.

In the following, we list some open problems for future work.

\begin{itemize}
    \item Can we improve the precision dependence of estimating the total variation distance between probability distributions? The currently best upper bound is $O(\sqrt{n}/\varepsilon^{2.5})$ due to  Montanaro~\cite{montanaro2015quantum}.

    \item Can we prove an $\Omega(\sqrt{n}/\varepsilon)$ lower bound on $\ell^1$-closeness testing?

    \item Can we relate the query  complexity of the $\ell^2$-closeness testing problem to the norm $\| p \|_2$ and $ \| q \|_2$ of probability distributions $p$ and $q$ as that for the classical algorithms?

    \item Can we close the gap between the upper bound and the lower bound of the $k$-wise uniformity testing problem?
\end{itemize}



\section*{Acknowledgment}

The work of Jingquan Luo and Lvzhou Li was supported by the National Science Foundation of China under Grant 62272492, the Guangdong Basic and Applied Basic Research Foundation under Grant 2020B1515020050 and the Guangdong Provincial Quantum Science Strategic Initiative under Grant GDZX2303007.
The work of Qisheng Wang was supported
by the Ministry of Education, Culture, Sports, Science and Technology (MEXT) Quantum Leap Flagship Program (MEXT Q-LEAP) under Grant JPMXS0120319794.

\bibliographystyle{IEEEtranN}
\bibliography{refs}

\newpage
\appendix
\section{Proof of \cref{cor:collision}}

In this section, we give the proof of \cref{cor:collision}, which is almost the same as the one in~\cite{kutin2005quantum}. 
The proof is based on the quantum polynomial method~\cite{beals2001quantum}.

We firstly give some intuition why \cref{the:collision} remains valid after two modifications. 
One modification is about the form of the oracle. Indeed, this would not change the fact that the acceptance probability of a quantum algorithm using $T$ queries can be written as a polynomial of degree at most $2T$. 
The other modification is to allow the algorithm to make mistakes on a negligible fraction of each kind of functions. A key step in the proof of the lower bound is to symmetrize the function, which in turn can be viewed as applying the algorithm to a \textit{random} input. 
As a result, making a negligible fraction of mistakes would only make little change to the expectation of the acceptance probability of the algorithm. Detailed proof is as follows.

Firstly, we introduce some necessary backgrounds, specifically, relating the query complexities of quantum algorithms to the degrees of polynonial funcitons. 

Let $\mathcal{F} = \sbra{n}^{\sbra{n}}$ be the set of all functions from $\sbra{n}$ to $\sbra{n}$. Given a function $f \in \mathcal{F} $, we assume that the inputs and outputs of the function are encoded using $\ceil{\log_2 n}$ qubits and the function is given as an oracle $O_f\colon |i\rangle | b \rangle \rightarrow | i \rangle | b \oplus f(i) \rangle$. 
Given an oracle $O_f$, a quantum query algorithm, which starts with state $\ket{0}$ and makes $T$ queries to $O_f$, can be represented as a sequence of operators as follows:
\begin{align*}
    U_0 \rightarrow O_f \rightarrow U_1 \rightarrow O_f \rightarrow \dots \rightarrow U_{T-1} \rightarrow O_f \rightarrow U_T \rightarrow \Pi,
\end{align*}
where $U_i, i = 0, \dots, T$ are arbitrary unitary operators and $\Pi$ are a projection operator. The acceptance probability is 
\begin{align}
\label{equ:pf}
    P(f) = \Abs{\Pi U_T O_f U_{T-1} \dots O_f U_1 O_f U_0 \ket{0}}^2.
\end{align}
Without loss of generality, for the collision problem, let the acceptance probability be the probability that the quantum algorithm asserts that $f$ is an $r$-to-one function. If the quantum algorithm correctly determine the type of function $f$ with probability at least $2/3$, we have $0 \leq P(f) \leq 1/3$ if $f$ is a one-to-one function and $2/3 \leq P(f) \leq 1$ otherwise.

For all $i, j \in \sbra{n}$, define $\delta_{i, j}\rbra{f}$ as follows:
\begin{align*}
    \delta_{i, j}\rbra{f} = \begin{cases}
        1 & \text{if}~f\rbra{i} = j, \\
        0 & otherwise. 
    \end{cases}
\end{align*}
Then, for any quantum algorithm that makes $T$ queries to $O_f$, the amplitude of each basis state of the final state can be written as a polynomial in $\delta_{i, j}\rbra{f}$ of degree at most $T$. Therefore the probability of acceptance $P\rbra{f}$ is a polynomial in $\delta_{i, j}\rbra{f}$ of degree at most $2T$. This relation was discovered by Beals~\textit{et~al.}~\cite{beals2001quantum}.

Let $S_n$ be the permutation group over $\sbra{n}$. For $\sigma, \tau \in S_n$, define $\Gamma^{\sigma}_{\tau} \colon \mathcal{F} \rightarrow \mathcal{F}$ by
\begin{align*}
    \Gamma^{\sigma}_{\tau} \rbra{f} = \tau \circ f \circ \sigma.
\end{align*}

Next we define a new class of functions which are $a$-to-one on part of the domain and $b$-to-one on the remaining. A triple $\rbra{m, a, b}$ of integers is said to be valid if $0 \leq m \leq n$, $a \mid m$ and $b \mid \rbra{n - m}$. For a valid triple, define $f_{m, a, b} \in \mathcal{F}$ as
\begin{align*}
    f_{m, a, b} \rbra{i} = \begin{cases}
        \ceil{i/a} & 1 \leq i \leq m, \\
        n - \floor{\rbra{n-i}/b} & m < i \leq n.
    \end{cases}
\end{align*}

For a valid triple $\rbra{m, a, b}$ and a polynomial $P\rbra{f}$ in $\delta_{i, j}\rbra{f}$, define $Q\rbra{m, a, b}$ as
\begin{align}
\label{equ:qmab}
    Q\rbra{m, a, b} = \mathbb{E}_{\sigma, \tau}\sbra{P\rbra{\Gamma^{\sigma}_{\tau} \rbra{f_{m, a, b}}}},
\end{align}
where $\sigma$ and $\tau$ are chosen uniformly at random on $S_n$.
The following lemma reduces the task of proving a lower bound on the degree of the polynomial $P(f)$ to proving a lower bound on the degree of $Q(m,a, b)$ by symmetrizing.

\begin{lemma}[{\cite[Lemma 2.2]{kutin2005quantum}}]
    Let $P\rbra{f}$ be the acceptance probability, as defined in \cref{equ:pf}, expressed as a polynomial in $\delta_{i, j}\rbra{f}$ of degree $d$. 
    Then $Q\rbra{m, a, b}$, as defined in \cref{equ:qmab}, is a polynomial of degree $d$ in $m$, $a$ and $b$.
\end{lemma}

The last component is a polynomial approximation theory due to  Paturi~\cite{paturi1992degree}, which is implictly given in the proof of the lower bound of the degree of the polynomial approximation to symmetric Boolean functions. We follow the statement given in~\cite{kutin2005quantum}.
\begin{theorem}[{\cite[Theorem 2.1]{kutin2005quantum}}]
\label{the:paturi}
    Let $q(\alpha) \in \mathbb{R}[\alpha]$ be a polynomial of degree $d$. Let $a, b$ be integers such that $a < b$, and let $\xi \in [a, b]$ be a real number. If
    \begin{enumerate}
        \item $\abs{q(i)} \leq c_1$ for all integers $i \in [a, b]$ for some constant $c_1 > 0$, and 
        \item $\abs{q(\ceil{\xi}) - q(\xi)} \geq c_2$ for some constant $c_2 > 0$,
    \end{enumerate}
    then 
    \begin{align*}
        d = \Omega\rbra*{\sqrt{\rbra*{\xi - a + 1}\rbra*{b - \xi +1}}}.
    \end{align*}
    In particular, $d = \Omega\rbra*{\sqrt{b - a}}$. And if $\xi = \frac{a + b}{2} + O(1)$, $d = \Omega\rbra*{b - a}$. 
\end{theorem}

Now we are ready to give the proof of \cref{cor:collision}. We restate the Corollary here.

\begin{corollary}[{\normalfont Restatement of \cref{cor:collision}}]
    Let $n > 0$ and $r \geq 2$ be integers such that $r \mid n$, and $f \colon [n] \rightarrow [n]$ be a function that is either one-to-one or $r$-to-one. Assume that the inputs and outputs of the functions are encoded using $\lceil{\log_2 n}\rceil$ qubits, and the function $f$ is given as an oracle $O_f\colon |i\rangle | b \rangle \rightarrow | i \rangle | b \oplus f(i) \rangle$. 
    Then any quantum algorithm that determines (with probability at least $2/3$) the type of the input function $f$ for $99\%$ one-to-one functions and $99\%$ $r$-to-one functions must make $\Omega((n/r)^{1/3}))$ queries to $O_f$.
\end{corollary}

\begin{proof}
    Let $\mathcal{A}$ be an algorithm that satisfies the requirements in the statement of the corollary and makes $T$ queries to $O_f$. Let $P(f)$ be the probability that $\mathcal{A}$ asserts that $f$ is an $r$-to-one function. From the above discussion, we know that $P(f)$ is a polynomial of degree at most $2T$ in $\delta_{i, j}(f)$. Let $Q(m, a, b)$ be formed from $P(f)$ as defined in \cref{equ:qmab} and let $d$ be the degree of $Q(m, a, b)$. Then we have $d \leq 2T$.

    For any $\sigma, \tau \in S_n$, we have $0 \leq P\rbra{\Gamma^{\sigma}_{\tau} \rbra{f_{m, a, b}}} \leq 1$. Then we have
    \begin{itemize}
        \item When $(m, a, b)$ is a valid triple, then $0 \leq Q(m, a, b) \leq 1$;
        
        \item When $ a = b = 1 $, $f_{m, 1, 1}$ is a one-to-one function. Then by assumption, for $99\%$ of $(\sigma, \tau) \in S_n \times S_n$, we have $ 0 \leq P\rbra{\Gamma^{\sigma}_{\tau} \rbra{f_{m, a, b}}} \leq 1 / 3 $, and for the remaining, $ 0 \leq P\rbra{\Gamma^{\sigma}_{\tau} \rbra{f_{m, a, b}}} \leq 1 $. Thus we have  $0 \leq Q(m, a, b) \leq 2/5$;

        \item When $a = b = r$ and $r \mid m$, $f_{m, r, r}$ is an $r$-to-one function. Then following a similar argument, we have $3/5 \leq Q(m, a, b) \leq 1 $.
    \end{itemize}

    Let $M = r\floor{n / 2r}$. The argument is splited into two cases based on whether $Q(M, 1, r) \geq 1/2$.
    \begin{itemize}
        \item Case 1: $Q(M, 1, r) \geq 1/2$. Let $g(x) = Q(M, 1, rx)$ and $k > 0$ be the least integer such that $\abs{g(k)} \geq 2$. We can find out such $k$ since $g(x)$ is a polynomial.
        Let $c = rk$. Then $g(1) - g(1/r) \geq 1/10$. By \cref{the:paturi}, we have 
        \begin{align}
            \label{equ:d1}
            d = \Omega(\sqrt{k}) = \Omega(\sqrt{c/r}).
        \end{align}

        Next consider function $h(i) = Q(n - ci, 1, c)$. For any integer $i$ such that $ 0 \leq i \leq \floor{n/c}$, the triple $(n-ci, i, c)$ is valid. Thus, $0 \leq h(i) \leq 1$. However, we have
        \begin{align}
        \label{equ:hi}
            \abs*{h\rbra*{\frac{n-M}{c}}} = \abs{Q(M, 1, c)} = \abs{g(k)} \geq 2.
        \end{align}
        Note that $\frac{n}{2c} \leq \frac{n-M}{c} \leq \frac{n+2r}{2c} \leq \frac{n}{2c} + 1$, by \cref{the:paturi}, we have 
        \begin{align}
        \label{equ:d2}
            d = \Omega(n/c).
        \end{align}

        \item Case 2: $Q(M, 1, r) < 1/2$. Let $g(x) = Q(M, rx, r)$. As before, let $k > 0$ be the least integer such that $\abs{g(k)} \geq 2$. Let $c = rk$. We have $g(1) - g(1/r) \geq 1/10$, which yields $d = \Omega(\sqrt{c})$ by \cref{the:paturi}.

        Next, consider function $h(i) = Q(ci, c, r)$. For any integer $i$ such that $ 0 \leq i \leq \floor{n/c}$, the triple $(ci, c, r)$ is valid as $r \mid n$ and $r \mid c$. Thus, $0 \leq h(i) \leq 1$. But, similar to \cref{equ:hi}, we have $\abs{h\rbra*{\frac{M}{c}}} \geq 2$. By \cref{the:paturi}, we have $d = \Omega(n/c)$.
    \end{itemize}

    In either cases, we have $d = \Omega(\sqrt{c/r})$ and $d = \Omega(n/c)$. Thus, we have $d^3 = \Omega(n / r)$, which completes the proof.
    
\end{proof}

\end{document}